\documentclass[sigconf]{acmart}

\AtBeginDocument{%
  \providecommand\BibTeX{{%
    \normalfont B\kern-0.5em{\scshape i\kern-0.25em b}\kern-0.8em\TeX}}}

\copyrightyear{2021}
\acmYear{2021}
\setcopyright{acmlicensed}\acmConference[FAccT '21]{Conference on Fairness, Accountability, and Transparency}{March 3--10, 2021}{Virtual Event, Canada}
\acmBooktitle{Conference on Fairness, Accountability, and Transparency (FAccT '21), March 3--10, 2021, Virtual Event, Canada}
\acmPrice{15.00}
\acmDOI{10.1145/3442188.3445882}
\acmISBN{978-1-4503-8309-7/21/03}

\usepackage{amsmath}               
\usepackage{amsthm}                
\usepackage{fancyhdr}              
\usepackage[shortlabels]{enumitem}
\PassOptionsToPackage{hyphens}{url}\usepackage{hyperref}
\usepackage[ruled, linesnumbered]{algorithm2e}
\usepackage[autostyle=true,german=quotes]{csquotes}%

\usepackage{bm}
 
\usepackage{graphicx}
\usepackage{color}

\setlist[enumerate]{leftmargin=*}  
\setlist[itemize]{leftmargin=*}        
\usepackage{thmtools,thm-restate}

\newtheorem{lemma}{Lemma}
\newtheorem{corollary}{Corollary}

\newcommand{\p}[1]{\left( #1 \right)}
\newcommand{\br}[1]{\left[ #1 \right]}

\newcommand{\NL}[0]{\ensuremath{N_L}}
\newcommand{\NH}[0]{\ensuremath{N_H}}
\newcommand{\nL}[0]{\ensuremath{n_L}}
\newcommand{\nH}[0]{\ensuremath{n_H}}
\newcommand{\nLv}[1]{\ensuremath{n_{L#1}}}
\newcommand{\nHv}[1]{\ensuremath{n_{H#1}}}
\newcommand{\rL}[0]{\ensuremath{r_L}}
\newcommand{\rH}[0]{\ensuremath{r_H}}
\newcommand{\RL}[0]{\ensuremath{R_L}}
\newcommand{\RH}[0]{\ensuremath{R_H}}
\newcommand{\V}[0]{\ensuremath{V}}
\newcommand{\bp}[0]{\ensuremath{b_p}}
\newcommand{\cost}[0]{\ensuremath{g}}

\newcommand{\priceL}[0]{\ensuremath{f_L}}
\newcommand{\priceH}[0]{\ensuremath{f_H}}

\newcommand{\cd}[0]{\cdot}

\usepackage{array}
\newcolumntype{C}[1]{>{\centering\let\newline\\\arraybackslash\hspace{0pt}}m{#1}}

\begin{document}

\title[How Externalities of Size Complicate Notions of Solidarity and Actuarial Fairness]{Better Together?\\ How Externalities of Size Complicate Notions of Solidarity and Actuarial Fairness}

\author{Kate Donahue}
\affiliation{%
  \institution{Cornell University}}
\email{kdonahue@cs.cornell.edu}

\author{Solon Barocas}
\affiliation{Microsoft Research and Cornell University}
\email{solon@microsoft.com}

\renewcommand{\shortauthors}{Donahue and Barocas}

\begin{abstract}
Consider a cost-sharing game with players of different costs: an example might be an insurance company calculating premiums for a population of mixed-risk individuals. Two natural and competing notions of fairness might be to a) charge each individual the same or b) charge each individual according to the cost that they bring to the pool. In the insurance literature, these approaches are referred to as \enquote{solidarity} and \enquote{actuarial fairness} and are commonly viewed as opposites. However, in insurance (and many other natural settings), the cost-sharing game also exhibits \emph{externalities of size}: all else being equal, larger groups have lower average cost. In the insurance case, we analyze model where costs strictly decreases with pooling due to a reduction in the variability of losses. In this paper, we explore how this complicates traditional understandings of fairness, drawing on literature in cooperative game theory.

First, we explore solidarity: we show that it is possible for both groups (high risk and low risk) to strictly benefit by joining an insurance pool where costs are evenly split, as opposed to being in separate risk pools. We build on this by producing a pricing scheme that maximally subsidizes the high risk group, while maintaining an incentive for lower risk people to stay in the insurance pool. Next, we demonstrate that with this new model, the price charged to each individual has to depend on the risk of other participants, making naive actuarial fairness inefficient. Furthermore, we prove that stable pricing schemes must be ones where players have the anti-social incentive desiring riskier partners, contradicting motivations for using actuarial fairness. Finally, we describe how these results relate to debates about fairness in machine learning and potential avenues for future research. 
\end{abstract}

\ccsdesc[500]{Theory of computation~Algorithmic game theory}
\ccsdesc[500]{Applied computing~Law}
\ccsdesc[300]{Applied computing~Economics}

\keywords{cooperative game theory, submodular cost function, fair cost sharing, insurance, solidarity, actuarial fairness}


\maketitle

\section{Introduction}

Imagine the following situation: an individual wishes to purchase car insurance and goes to her local insurance company. The insurance company combines her with other people insuring cars of the same value. Within the insurance pool, some people are at low risk of being in an accident and some people at a high risk. The insurance company calculates a total amount of premiums that it needs to collect from the entire set of people: higher risk individuals contribute more to this total than lower risk individuals. How should total premiums be divided into prices for each individual?

One suggestion is to charge people proportional to their risk level, requiring those who bring more cost to the insurance pool to pay more. However, there may be reasons to avoid this approach: what if individuals aren't responsible for their risk status? Consider a case where lower risk people live in suburbs with little traffic, while higher risk people tend to live in dense urban areas --- and consider that these patterns may be due to economic inequality or racial discrimination. In this case, there might be normative reasons to favor dividing total premiums evenly among policy-holders. 

The question of how to divide insurance premiums has been debated extensively in the insurance literature, as we review in Section \ref{literature}. In general, this literature casts these two choices --- proportional pricing (related to the idea of actuarial fairness) or equal pricing (related to the idea of solidarity) --- as strict opposites. However, we show that this idea assumes a fairly simplistic form of premium calculation. In this work, we will explore the impact of a slightly more realistic model: one involving \emph{externalities of size}. 

Simply put, a model exhibits externalities of size when, all else being equal, larger groups have lower average cost. For the insurance example, we will demonstrate in Section \ref{exsize} how insurance pools help reduce variance and thereby bring down total costs in ways that can benefit all members of the pool. However, our results are not limited solely to the insurance case: all of theoretical proofs are given for general cost functions and would apply to any systems exhibiting similar properties.

Externalities of size complicate current debates in the Fairness, Accountability, and Transparency (FAccT) community: when an individual can both be responsible for increasing costs (due to being higher risk) and decreasing costs (due to externalities of size), how should we reason about which price to charge her? Cooperative game theory has analyzed this question, but from a different angle: the questions there focus on how to charge more costly individuals an appropriately greater amount, in keeping with actuarial fairness, whereas the FAccT community is more interested in questions of equity and equality, as championed by advocates of solidarity. 

As mentioned before, our results hold for general cost-sharing situations, but throughout this work we will focus on the application area of insurance. There are two main reasons for this focus: firstly, there is a rich literature in philosophy, law, sociology, and history debating the meaning of fairness in insurance. We draw on this literature to motivate our analysis and use our results to challenge some of the entrenched thinking and suggest potential avenues of future research. The second motivation for focusing on insurance is that there already exists a natural model for premium calculation with externalities of size. In Section \ref{literature}, we describe this model and discuss how it satisfies our desired properties. In this work, we will explicitly \emph{not} analyze health insurance, which is so essential to well-being that many societies treat it as a right.

In later sections, our analysis will show that, due to externalities of size, the distinction between actuarial fairness and solidarity quickly blurs. In Section \ref{solidarity}, we are motivated by cases where externalities of size may make it easier to accomplish solidaristic goals. First, we show that, for some situations, it is possible that evenly splitting premiums ends up producing prices that are strictly lower than what either the low or high risk participants would get in homogeneous pools. We additionally define and analyze a pricing method that allows us to maximally subisidize the higher risk participants while still remaining stable against defections. 

In Section \ref{Actfair}, we turn to actuarial fairness. First, we show that certain properties of actuarial fairness are impossible to simultaneously achieve with externalities of size. For example, a pricing scheme that is indifferent to risk levels of other participants is also one that is inefficient. Additionally, we show that a pricing scheme that is both stable and efficient is one where players have the anti-social incentive of wishing that other players have higher risk. We discuss how these results might have implications for debates in fairness within the insurance literature and more broadly. 

\section{Motivating literature}\label{literature}
\subsection{Debates over fairness in insurance: solidarity and actuarial fairness}\label{insurance}

The history of insurance and premium calculation in the Western world is a surprisingly rich and fascinating one. Historically, harmful events were not understood as chance occurrences; rather, \enquote{accidents and diseases were seen as a punishment for objectionable deeds, known only to God and the sinner} \cite{barry2019insurance}. Spreading the burdens of these misfortunes across a larger group of people made no sense under this worldview, as someone's chance of experiencing a misfortune was thought to be entirely under their control \cite{bernstein1996against}. It was only during the industrial revolution that these views started to change, as workplace accidents began to be seen as not necessarily the fault of workers or managers, but as \enquote{inherent to industrialization itself,} leading some to conclude \enquote{that the old rule of responsibility was obsolete}  \cite{barry2019insurance}. 

Instead, insurance allowed multiple people to pool and manage risks communally. A rare but disastrous accident, like a house fire, could be insured against for a fixed sum, with the common result that all participants benefited. In this respect, \citet{lehtonen2011forms} write that insurance should be understood as \enquote{a central yet often inconspicuous infrastructure supporting the Western way of life}. \citet{horan2011actuarial} has likewise argued that insurance played a critical role in helping reshape American life in the post-World War II era. But if insurance increased welfare, it also raised questions of how its costs should be divided into premiums.

The most natural choice might be to simply calculate the necessary total amount of money required to insure a group of people and then divide the cost evenly between them. But as actuarial science grew more sophisticated, it became apparent that certain individuals had a greater chance of suffering a loss than others. Is it really right that they should pay the same amount? 

\citet{arrow1978uncertainty} would answer \enquote{no}: this seminal work defined \emph{actuarially fair} pricing as that in which each participant in the insurance pool pays their expected costs. One motivation for using actuarial fairness is that failing to do so could lead to adverse selection. Adverse selection is the phenomenon where a pricing scheme that charges high and low risk individuals the same amount ends up incentivizing low risk individuals to leave, either forming their own (cheaper) pool under a different insurance company or to leave the insurance market altogether. The remaining insurance pool has higher average risk and is thus more expensive, which may convince the remaining lower risk individuals to leave as well, starting a potential \enquote{death spiral.} (Empirically, adverse selection has proven a less serious threat to insurance markets \cite{siegelman2003adverse} than this theory would suggest.) Beyond adverse selection, we may object to equal pricing because it invites \enquote{moral hazard}: unless policyholders face some financial consequences for doing so, they may engage in needlessly risky behavior. The concerns with moral hazard are both practical and normative. Most immediately, the problem with such behavior is that it raises the total cost that must be borne by the risk pool. But it also provokes objections on the ground that it is nor fair for an individual to saddle other people with the costs of their reckless behavior \cite{mas1995microeconomic}. For this reason, actuarial fairness is sometimes justified on the basis that \enquote{it is unfair for some individuals to bear costs stemming from the actions of others} \cite{landes2015fair}. 

Of course, this view isn't universally held: a different view is \emph{solidarity}, which holds that groups with different risks should nevertheless still pay equal premiums. This view holds that charging different premiums runs counter to the purpose of insurance, which is to spread the costs of events outside an individual's control across the collective. In contrast, actuarial fairness is ultimately indifferent to whether someone is responsible for their risk status: the danger of adverse selection would seem to compel insurers to charge higher-risk policyholders a higher price even if they are at higher relative risk for no fault of their own. A solidaristic view of insurance holds that its function is to help to compensate those who happen to experience misfortune, but also those who happen to be at greater risk of misfortune for reasons outside their control. For example, consider a person who can only afford an older car with a higher chance of being totalled in the event of an accident. We may not wish to charge her a higher premium because the fact that she happens to own an older car is explained by her income and wealth, properties over which she may have limited control. In contrast, we might feel justified in charging her a higher price if her higher risk of totalling her car is the result of reckless driving, an activity over which she does has control. In this sense, solidarity is about helping individuals find ways to cover the costs of each other's misfortunes, including the misfortune of being at greater risk for certain bad outcomes. 

\subsection{Game theory}
Questions of how to divide costs among participants fall into the realm of \emph{cooperative game theory}. This area studies situations where players form coalitions that produce benefit or cost which must be divided among the participants. Key questions center on which coalitions utility-maximizing players have incentives to join and which coalitions are \enquote{stable} against defecting groups of players. 

The seminal works of \citet{shapley1971cores} and \citet{bondareva1963some} analyze certain classes of cost functions and give guarantees for stable cost-sharing schemes. Later works, such as \citet{csoka2012note} and \citet{Balog2014Properties}, specifically use cooperative game theory to analyze situations where risk is shared among coalitions or groups of different actors. However, their work differs from ours because we are considering allocating cost, rather than risk: our cost function is based off of probabilistic factors like the risks $\{r_i\}$, but theirs is a perfectly deterministic function. Closer works to ours include \citet{elastic}, which analyzes \enquote{elastic} cost functions, and \citet{guo2013fair}, an applied example around allocating resources across call centers (which has a cost function structurally similar to ours).

Recent papers have similarly drawn connections between the game theory and the FAccT communities. For example, \citet{hu2020fair} demonstrates that applying strict notions of fairness can reduce the welfare of both relevant groups. \citet{kasy2020fairness} examines similar tensions, showing multiple examples where changes that increase fairness simultaneously decrease equity and welfare. Finally, \citet{finocchiarofairness} provides an overview of points of contact between the mechanism design and FAccT communities. 

Part of this paper's goal is translational: to use the tools developed in cooperative game theory to shed light on avenues of research in the insurance and fairness literature that may not have been considered otherwise. But part of our analysis is fundamentally different from the goals most game theory focuses on. Specifically, much of game theory relies on the idea of \enquote{fairness} as meaning people who bring more cost to the group should pay a larger share of the total cost. For example, \citet{elastic} explores multiple variants of pricing schemes that attempt to enforce higher prices for those who contribute higher costs. However, as discussed in the end of Section \ref{insurance}, in many realistic cases we might say that it is \enquote{fair} to \emph{not} charge a higher-risk person more. This type of analysis is much less explored in the game theory literature.

\section{Externalities of size}\label{exsize}
In this section, we discuss our definition of externalities of size. First, we define the term and give examples of real-world phenomena exhibiting it. Next, in the insurance case, we demonstrate that the literature discussed in Section \ref{literature} assumes a simpler, somewhat unrealistic model of insurance premium calculation. However, we show that a more realistic model, called insolvency-based premium calculation, exhibits externalities of size.  

\subsection{Definition and examples}
In this paper, we consider cases where the total cost associated with a group is a function of the identity of people in that group. In this case, we can consider a set function $\cost(\cd)$ that takes in a set of players $S$ and returns a real number. One very simple example of a $\cost(\cd)$ function is a linear cost-sharing game where the cost of a set of players $S$ is the sum of the costs of the individual players:
$$\cost(S) = \sum_{i \in S}\cost(\{i\}) $$
In this work, we argue that, in many cases, this model is too simplified: it ignores the relevant concept of \emph{externalities of size}. Specifically, this means that the total costs involved in a set of players $S$ is strictly less than the sum of the costs of each player:
$$\cost(S) < \sum_{i \in S}\cost(\{i\}) $$
There are multiple examples of real-world phenomena exhibiting externalities of size:
\begin{itemize}
    \item A multi-state coffee chain is cheaper to operate, in aggregate, than each of the separate locations would be to run individually. 
    \item A restaurant can cook 50 dishes with greater ease and less expense than 50 separate cooks each making one dish. 
    \item For delivery to a large, contiguous area, a single delivery service can be operated with less total cost than multiple delivery services within the same area. 
\end{itemize}
Note that in the cases above, the total savings hold even if individual elements are highly unequal in their contribution to total costs:
\begin{itemize}
    \item For the coffee chain, it may be the case that stores in high-rent cities are much more expensive to operate than ones in small towns. 
    \item Certain dishes may be much more expensive, in time and ingredients, than other dishes. 
    \item Deliveries to remote areas may be more time-consuming and expensive than typical deliveries. 
\end{itemize}
Besides these applied examples, there has also been theoretical work analyzing cost-sharing in a situation with unequal costs but externalities of size. For example, \citet{herzog1997sharing} analyzes cost-sharing in computer networking. In this case, if multiple participants build a network together, average costs per person shrink because users can split costs for portions of the network that they share. 

\subsection{Insolvency-based premiums}
Next, we will analyze the cost function used in insurance literature. For simplicity, we will assume all insurance policies cover the same value $\V$ with a binary loss of probability $r_i$ for individual $i$. 

\citet{meyers2018enacting} describes definitions of actuarial fairness from multiple papers and textbooks, which generally agree that \enquote{a fairly priced insurance policy is one in which the insurance premium is equal to the expected value of the promised insurance payment}. This implies that an actuarially fair pricing scheme would collect total premiums $C$ for a set of players $S$ according to the equation below:
$$C = \sum_{i \in S}\V \cd r_i$$
Those holding a solidaristic view of insurance would disagree on how the total premiums should be divided, but would agree that $C$ corresponds to the correct total value. This cost function is linear, which means that there is a sharp tension between what different people pay. If one person pays less than his expected cost, a different person must pay more than her expected cost in order for the total amount of premiums collected to sum up to the amount needed. 

However, though it does not appear to be discussed in the fairness in insurance literature, there does already exist a model of insurance premiums that exhibits externalities of size: \emph{insolvency-based premiums} \cite{olivieri2015introduction}. 

To motivate this model, it helps to consider what insurance represents. In exchange for paying premiums, a policyholders receives a promise from the insurance company to repay the costs that she suffers in the event of a loss. A bad situation would be if the total value of claims in a given year exceeded the total value of premiums paid. If a \enquote{shortfall} happens, some policyholders may go uncompensated. What is the probability of a shortfall? Denote the random variable describing the total value of claims in a year by $X$. Using the pricing scheme referenced in the insurance literature above, a shortfall occurs whenever $X > \mathbb{E}[X]$, which for a symmetric distribution like our model occurs with 50\% chance! This probability might be unacceptably high --- and another drawback is that it does not depend on the size of an insurance pool: a large pool would be as likely to experience a shortfall as a small one. 

An alternate pricing scheme is called \emph{insolvency-based pricing}. In this model, a premium is collected so that the probability of a shortfall is no more than some fixed $p$: 
$$P(X > C_p) \leq p$$
The total amount of premiums $C_p$ is a function of the probability $p$, which could be viewed as an external requirement, potentially imposed by a regulator, to require that the company maintains some level of financial stability. In this way, multiple insurance companies are assumed to share the same cost function and differ only in their composition of policyholders. 

\subsection{Variance reduction}
For the model of insurance losses we are considering, it is possible to perfectly calculate $C_p$. First, we will assume that policyholders come in two types: low risk and high risk. There are $\nL$ low risk policyholders with risk $\rL$, and $\nH$ high risk policyholders with risk $\rH>\rL$. As before, each is insuring a good of value $V$. 

The total number of claims coming from the low risk participants is distributed according to a binomial distribution with parameters $\nL$ and $\rL$. It vastly simplifies our analysis to approximate this distribution as a normal distribution with expected value $\nL \cd \rL$ and standard deviation $\sqrt{\nL \cd \rL \cd (1-\rL)}$: when the $\nL$ is fairly large, which is common in insurance applications, this is a good approximation.  Similarly, we can approximate the distribution describing the number of claims from the high risk group as a normal distribution with expected value $\nH \cd \rH$ and standard deviation $\sqrt{\nH \cd \rH \cd (1-\rH)}$. The total number of claims in a combined insurance pool is then a normal distribution with mean $\nL \cd \rL +\nH \cd \rH$ and standard deviation $\sqrt{\nL \cd \rL \cd (1-\rL) + \nH \cd \rH \cd (1-\rH)}$. Because each claim has an identical value $\V$, the total value of claims simply scales the mean and standard deviation by $\V$. 

The benefit of using a normal approximation becomes clear: it is extremely straightforward to calculate $C_p$. Given a normal distribution with mean $\mu$ and standard deviation $\sigma$, $C_p$ in the form $\mu + b_p \cd \sigma$ produces the desired premium, where $b_p$ is a constant that depends on $p$ but not $\mu$ or $\sigma$ values. 

These results give the amount of money that insolvency-based premium calculation would need to collect for given population: $\cost(\{\nL, \nH\}, \{\rL, \rH\}) = $
$$\V \p{\rL \cd \nL + \rH \cd \nH + \bp \cd \sqrt{\nL \cdot \rL \cd (1-\rL) + \nH \cd \rH \cd (1-\rH)}}$$
We have slightly abused notation by allowing $\cost(\cd)$ to stand both for the function on $\nL, \nH$ and a set function.

Note that pooling always reduces costs: 
$$ \cost(\{\nL, 0\}, \{\rL, 0\}) + \cost(\{0, \nH\}, \{0, \rH\}) > \cost(\{\nL, \nH\}, \{\rL, \rH\})$$
The expected value component $\V \p{\rL \cd \nL + \rH \cd \nH}$ is the same on both sides. However, the pooled insurance group has lower cost through reduced variance: 
$$\V \cd \bp \p{\sqrt{\nL \cdot \rL \cd (1-\rL)} +  \sqrt{\nH \cd \rH \cd (1-\rH)}}$$
$$> \V \cd \bp \p{\sqrt{\nL \cdot \rL \cd (1-\rL) + \nH \cd \rH \cd (1-\rH)}}$$
The standard deviation shrinks, which reduces the total amount of money that needs to be collected. 

The next lemma strengthens and formalizes these results by showing that the cost function is submodular. 

\begin{restatable}{lemma}{submod}
\label{submod}
For insolvency-based premium calculation, the cost function $\cost(\cdot)$ is submodular. That is, for all sets of players $S$ and $T$, 
$$\cost(S) + \cost(T) \geq \cost(S \cup T) + \cost(S \cap T)$$
The function is strictly submodular: that is, the inequality is strict whenever it is the case that S is not a subset of T and T is not a subset of S, so they both have non-overlapping portions. 
\end{restatable}
This proof is given in Appendix \ref{App:proofs}. This property is one we will rely on in later sections to demonstrate overall cost-savings. 

\subsection{Model characteristics}

This section contains some additional notes on possible objections to the insolvency-based premium model. First, this model as written assumes that the insurance company doesn't have any stockpile of money it could use as a cushion in case costs are unexpectedly high. In reality, insurance companies would almost surely have such a financial cushion, but it also seems certain that would want to be compensated financially for the opportunity cost of not using this money in other ways. A reasonable solution would be to have the policyholders pay enough in premiums to at least compensate the insurance company for lost interest on the insurance stockpile: such a scheme would produce a premium of the same form, but with a constant in front of the $\bp$.

Secondly, the $\cost(\cd)$ function as written produces an average premium that is strictly higher than the expected value of losses. Some might object to this: why would someone pay more than their expected loss? The key is that an individual purchasing insurance is purchasing a \emph{reduction in the variability of their costs}. For example, consider an individual purchasing insurance in the model above. Without insurance, her loss in each time period would have expectation $r \cd \V$ and standard deviation $\V \cd \sqrt{r \cd (1-r)}$. In the case that $\V$ is large, this standard deviation could be quite large: she might need to establish a costly financial cushion to handle this uncertainty in her losses. If she purchases insurance, her loss each time period is equal to her premium $C_A$---with 0 standard deviation. The ability to have consistency in her losses might be very valuable to her, which explains why she would be willing to pay a premium $C_A$ that is strictly greater than her expected loss\footnote{One clarifying point: it is important to distinguish between the reduction in variance a policyholder purchases when she buys insurance and the reduction in variance that occurs in an insurance pool when more people are added. The first case is a motivating reason why people purchase insurance, but the second case is a reason why larger insurance pools are helpful (and is a major focus of this paper).}.

\section{Motivating Example}\label{motivate}

\begin{table}[]
\begin{tabular}{|c|c|c|c|}
\hline
                                                                                 & \textbf{Total} & \textbf{Low risk} & \textbf{High risk} \\ \hline
\textbf{Separate pools}                                                          & \$22,500           & \$20                  & \$25                \\ \hline
\textbf{\begin{tabular}[c]{@{}c@{}}Pooled: \\ even-split pricing\end{tabular}}        & \$22,500          & \$22.50                  & \$22.50                \\ \hline
\textbf{\begin{tabular}[c]{@{}c@{}}Pooled: \\ proportional pricing\end{tabular}} & \$22,500        & \$20                  & \$25                \\ \hline
\end{tabular}
\caption{Example of pricing with expected-value premiums, for low risk group (2\% chance of suffering loss) and high risk groups (2.5\% chance).}
\label{tab:exp}
    \vspace{-7mm}
\end{table}

\begin{table}[]
\begin{tabular}{|c|c|c|c|}
\hline
                                                                                 & \textbf{Total} & \textbf{Low risk} & \textbf{High risk} \\ \hline
\textbf{Separate pools}                                                          & \$35,741       & \$32.52               & \$38.96           \\ \hline
\textbf{\begin{tabular}[c]{@{}c@{}}Pooled: \\  even-split pricing\end{tabular}}        & \$31,878         & \$31.88               & \$31.88             \\ \hline
\textbf{\begin{tabular}[c]{@{}c@{}}Pooled: \\ proportional pricing\end{tabular}} & \$31,878           & \$28.36               & \$35.40           \\ \hline
\end{tabular}
\caption{Example with solvency-based premium calculation, with same risk levels as above.}
\label{tab:insolv}
    \vspace{-7mm}
\end{table}

\begin{table}[]
\begin{tabular}{|c|c|c|c|}
\hline
                                                                                 & \textbf{Total} & \textbf{Low risk} & \textbf{High risk} \\ \hline
\textbf{Separate pools}                                                          & \$45,024        & \$32.52               & \$57.53            \\ \hline
\textbf{\begin{tabular}[c]{@{}c@{}}Pooled: \\  even-split pricing\end{tabular}}        & \$40,770          & \$40.77               & \$40.77              \\ \hline
\textbf{\begin{tabular}[c]{@{}c@{}}Pooled: \\ proportional pricing\end{tabular}} & \$40,770           & \$27.28               & \$54.26             \\ \hline
\end{tabular}
\caption{Example with insolvency-based premium calculation, where risk level is the same for low risk group, but 4\% for the high risk group.}
\label{tab:insolv_riskier}
\vspace{-8mm}
\end{table}

Consider a scenario with a single insurable loss with value \$1,000: for example, consider insuring a car against total loss. There are 1,000 possible policy-holders: 500 of them have a 2\% chance of suffering the loss (low risk) and 500 have a 2.5\% chance of suffering the loss (high risk).

Table \ref{tab:exp} describes this scenario with expected-value premiums as assumed in much of the insurance literature. The first row describes the premiums collected if low risk and high risk individuals are in separate insurance pools,  potentially at different companies. The second and third rows consider two different ways of pricing premiums when all of the individuals are in the same insurance pool.  Note that the total amount of premiums stays the same in all three situations, which is a feature of the expected value pricing scheme. The second row has even-split pricing, which might be consider solidaristic. Note that the low risk policy-holders pay strictly more and the high risk policy-holders pay strictly less than if they were in separate insurance pools: this is a necessary property of any solidaristic pricing scheme with this model. The third row describes proportional pricing, which might be viewed as an actuarially-fair pricing scheme. Note that both types of policy-holders pay amounts proportional to their risk: in this case, they neither benefit nor are hurt from being in an insurance pool together. 

Next, Table \ref{tab:insolv} analyzes the same scenario, but under the insolvency-based premium calculation with externalities of size as discussed previously. This example uses $\bp = 2$ to give a 2.27\% chance of insolvency. Note that the total cost of insuring all of the individuals is lower when they are pooled together, as opposed to being in separate homogeneous pools. The second row reflects equal pricing: both low and high risk policy-holders see their costs decrease! From the perspective of the low risk policyholders, the decrease in overall costs due to externalities of size outweighs the costs of pooling with a higher risk group. The last row contains the values for proportional pricing, which shows that both types of policy-holders see strictly lower costs than they would get alone. Later sections will describe exactly how proportional pricing is calculated in this model, and will also demonstrate that there always exist pricing schemes where both low and high risk policyholders strictly benefit from being pooled together, regardless of their risk levels.

Finally, Table \ref{tab:insolv_riskier} also shows an example of insolvency-based premium calculation, but when high risk players are even riskier: they each have a 4\% chance of a loss. As expected, the total amount of money that needs to be collected in premiums is higher. Here, even-split pricing still lowers the high-risk policyholder's price, but it increases the low risk policyholder's price compared to being in a separate insurance pool. Surprisingly, under proportional pricing the low risk policyholders pays \emph{less} than when did when the low risk policyholder's risk was lower, at 2.5\%. Later results will show that this kind of anti-social incentive for other policyholders to have higher risks is a necessary feature of pricing policies like this. 

\section{Model and assumptions}
\subsection{Model and terminology}
\textbf{Insurance model}: We consider the case where there are two types of people, low risk or high risk. There are a total of $\NL$ low risk players and $\NH$ high risk players.  Each person is considering purchasing insurance that would cover them completely in case of a specific loss of value \V, where such a loss can happen either zero times or once during the insurance time period. The low risk policyholder has probability $\rL$ of suffering such a loss, while the high risk policyholders has probability $\rH$: such probabilities are perfectly known to all participants, as well as the insurance companies. We will assume that $\rL < \rH < 0.5$, that all losses occur independently of each other, and that the insurance pool is using insolvency-based premium calculation.

\noindent \textbf{General model (beyond insurance)}: We again assume there are two types of people, each with a low or high cost associated: $\rL \in \mathbb{R}_{\geq0}, \rH \in \mathbb{R}_{\geq0}$. The total cost generated by a set of $\nL$ low risk players and $\nH$ high risk players is $\cost(\{\nL, \nH\}, \{\rL, \rH\})$. The cost is monotone and high risk players are more costly: for any $\nL, \nH, \rL, \rH$, the increase in $\cost(\cd)$ from adding a single low risk player is strictly lower than the increase in $\cost(\cd)$ from adding a single high risk player. 

We will assume the function is strictly submodular in $\nL, \nH$, as in Lemma \ref{submod}, and also that it is continuous in $\rL \in [0, \infty), \rH \in [0, \infty)$. Additionally, we will assume that a player with 0 cost contributes nothing to the total cost, implying that  $\cost(\{\nL, \nH\}, \{0, 0\}) = 0$: a pool whose members have 0 cost produces 0 total cost.

\noindent \textbf{Terminology:} We will often refer to the participants in the pooled activity as players, agents, or policyholders. Sometimes we will refer to the groups they form as pools or coalitions. A collection of coalitions is a coalition structure. A coalition structure is \emph{core-stable} if there does not exist a group of player $S$ so that each player would strictly prefer to be in $S$ as opposed to being in their present pool. We will sometimes use the notation $\pi(\nL, \nH)$ to refer to a coalition with $\nL$ low risk players and $\nH$ high risk players. The pool containing all of the players will be called the \enquote{grand coalition} and can also be written $\pi(\NL, \NH)$. We will use $\priceL(\cd), \priceH(\cd)$ to refer to the prices charged to low and high risk players respectively. 

\subsection{Technical assumptions}

One common assumption in arguments about insurance is moral hazard, which relates to the incentives people have to change their risks (or costs, in the general model). By this argument, charging a high-risk individual more will incentivize them to reduce their risks: failing to charge them more would incentivize riskier behavior. This increased cost is then passed on to the entire pool. In this paper, we do \emph{not} assume that the premium an individual is charged influences their risk level. This is not just an assumption of convenience, but is motivated by consideration of what we should consider \enquote{voluntary}. Some actions that could reduce risk of a loss could be extremely costly --- and some people will be better positioned to incur these costs than others. For example, supplemental driving lessons, beyond those required by the state, likely reduce the risk of an accident, but they impose additional costs on drivers that not everyone will be able to incur. Within this paper, we are assuming that an individual's risk is either immutable or that the cost of changing the risk is prohibitively expensive.

We will also assume that pricing must be \emph{efficient}: the total amount paid by all members in a pool must sum up to the total amount required as calculated by the cost function. There are situations where this might not be true, where there are savings, for example. However, omitting the efficiency assumption makes it very hard to say anything technical: when the price can be completely unrelated to the cost needed, it is not possible to guarantee anything about the prices. We will also assume that prices can depend only on the risk of the individual, which implies that all individuals with the same risk must have the same price.

\subsection{Normative assumptions}
In this work, we necessarily make a number of normative assumptions alongside our technical assumptions \cite{cooper2020normative}. We focus our attention on insurance products that provide value to policyholders and to the world (excluding insurance of harmful or objectionable activities). We will assume that being denied insurance has a negative impact on their life, either through direct loss of insurance or loss of essential goods afforded by insurance. For example, someone denied home-owners insurance loses both the insurance as well as potentially the ability to get a mortgage \cite{heimer2002insuring}. As mentioned earlier, we explicitly do \emph{not} analyze health insurance, which is so essential to well-being that many societies treat it as a right.

 We thus subscribe to the normative belief that a pricing scheme that induces participation in the insurance market is preferable to other arrangements in which agents might have rational incentives to opt out completely. We further assume that arrangements that maximize welfare (by reducing overall costs while providing the same level of value to policyholders) are normatively preferable to others: in particular, this means that we prefer cases where all policyholders are in the same insurance pool (the grand coalition), when possible. And we take as a given, for reasons described earlier, the desirability of an arrangement that maximally subsidizes high-risk agents, while ensuring the stability of the grand coalition, as this serves the twin normative goals of reducing inequality and obtaining the collective welfare gains from economies of size. While individual insurers might be motivated to adopt such a scheme with profit-maximization in mind (because such a scheme would help ensure that agents are not convinced to join an insurance pool run by a competitor), this is not our primary concern.

We also note that setting premiums is only one of many possible mechanisms that could be used to achieve policy goals. For example, even if our analysis says that a certain pricing scheme is \enquote{infeasible,} it could still be the case that, for example, tax subsidies and redistribution could achieve an equivalent result. In fact, our analysis will reveal when it is necessary to consider such alternatives.

\section{Solidarity under externalities of size}\label{solidarity}

So far, we have introduced a model for calculating the total costs created by a group of individuals. Here, we take a solidaristic perspective in how we might divide this cost.  

The first section implements \enquote{even split} pricing, where both players pay the same amount. This approach most closely matches what advocates of solidarity might suggest. Interestingly, we show that sometimes even-split pricing can strictly benefit \emph{both} the low risk and high risk players financially. However, there are also situations where even-split pricing is too aggressive and ends up hurting both the low risk and high risk players.

Next, the second section explores a more flexible notion of fairness: one where we minimize the cost paid by the \emph{more} expensive high risk participants, while maintaining stability. This pricing scheme might be useful in cases where we wish to financially support high risk players as much as we can without causing low risk players to wish to defect. Finally, the last section applies these results specifically to the insurance premium case and discusses the implications for the fairness debate.

\subsection{Even-split price}
Consider the \emph{even-split} pricing scheme defined below. 

\begin{definition}\label{def:evensplit}
With \emph{even-split} pricing, both the high and low risk participants pay the same amount: $\frac{\cost(\{\nL,\nH \}, \{\rL, \rH\})}{\nL + \nH}$. 
\end{definition}

As mentioned before, this pricing scheme follows a natural philosophy of solidarity: all participants should pay the same amount. For a cost function $\cost(\cd)$ that is linear, as is assumed in much of the literature on insurance, such a pricing scheme must strictly hurt the low risk players and strictly help the high risk players. For a submodular cost function, the result is more complicated. The lemma below describes a situation where even-split pricing makes the grand coalition core-stable: that is, no subgroup of players wishes to deviate and form their own pool. 

\begin{restatable}{lemma}{evenstab}
\label{evenstab}
With even-split pricing, if the below inequality is satisfied, then the grand coalition ($\pi(\NL, \NH)$) is core-stable.
$$\frac{\cost(\{\NL,\NH \}, \{\rL, \rH\})}{\NL + \NH} < \frac{\cost(\{\NL,0 \}, \{\rL,0\})}{\NL }$$
\end{restatable}
The proof is presented in Appendix \ref{App:proofs}.

It is worth pausing to realize why Lemma \ref{evenstab} is useful. The inequality states the grand coalition is stable against the deviation where all the low risk players form their own group in $\pi(\NL, 0)$. The lemma tells us that the inequality implies that something stronger: that the grand coalition is stable against deviations from \emph{every other possible combination} of low risk and high risk players $\pi(\nL, \nH)$.

Overall, this result suggests that there exist situations where both the low risk players and the high risk players benefit financially from solidarity as implemented in even-split pricing. To understand why this happen, it helps to think of there being two countervailing forces: one is that each additional person increases total costs, but the other is that, through submodularity, they may produce cost savings. When the cost savings outweigh the cost increases, it may be possible for both groups to benefit from even-split pricing. 

However, this is not always the case. Corollary \ref{negative} states a more pessimistic implication: there exist cases where even-split pricing financially hurts the high risk participants that it aims to help. 

\begin{corollary}\label{negative}
If the inequality in Lemma \ref{evenstab} does not hold, then the grand coalition will not be stable: the low risk players have an incentive to defect to $\pi(\NL, 0)$, where they will pay a lower amount. 
\end{corollary}
The corollary implies that if even-split pricing is implemented in this situation, the low risk players will leave the grand coalition, leaving the high risk players in the coalition $\pi(0, \NH)$ and paying a price $\frac{\cost(\{0,\NH \}, \{0, \rH\})}{\NH}$. In this way, attempting to enforce solidarity can make both groups worse off. 

This result matches the intuition developed from other analysis. For example, \citet{akerlof1978market} describes how information asymmetry could cause a market to fall apart, even when there were willing sellers and buyers. Here, even-split pricing mimics information asymmetry because it is impossible to distinguish between the low risk and high risk participants. The grand coalition pool falls apart even though it is possible to produce a pricing scheme where both low risk and high risk players benefit from being combined. Similarly, this result matches the analysis in works like \citet{kasy2020fairness} which demonstrated that, in certain situations, enforcing fairness can reduce welfare for both groups.

\subsection{Max-subsidy}

In this section, we explore the following pricing scheme that aims to help subsidize high risk player's price as much as possible, while still ensuring low risk players have an incentive to participate in the grand coalition. As mentioned before, this is a more flexible notion of solidarity than even-split pricing: we will explore it as a complement to the results we derived there. First, we define the pricing scheme. 

\begin{definition}\label{def:maxsub}
\emph{Max-subsidy pricing} follows this pricing policy: 
$$\priceL(\{\nL,\nH \}, \{\rL, \rH\}) = \frac{\cost(\{\nL,0 \}, \{\rL, 0\})}{\nL}$$
$$\priceH(\{\nL,\nH \}, \{\rL, \rH\}) = \frac{\cost(\{\nL,\nH \}, \{\rL, \rH\}) - \cost(\{\nL,0 \}, \{\rL, 0\})}{\nH}$$
\end{definition}
First, we will show that with this pricing scheme, high risk players most prefer being in the grand coalition, while low risk players most prefer being with as many other low risk players as possible (and do not care about the presence of high risk players). The proof is given in Appendix \ref{App:proofs}. The corollary shows that this implies that the grand coalition is core-stable.

\begin{restatable}{lemma}{maxsubpref}
\label{maxsubpref}
For max-subsidy pricing, $\priceL(\cd)$ is decreasing in $\nL$ and constant in $\nH$, $\priceH(\cd)$ is decreasing in both $\nL$ and $\nH$. 
\end{restatable}

\begin{corollary}
Assume a submodular cost function $\cost(\cd)$ with max-subsidy pricing. The \enquote{grand coalition} where all players are in the same insurance pool is core-stable. 
\end{corollary}

\begin{proof}
Showing that the grand coalition is core-stable means that there does not exist a group of players $S$ whose members all strictly prefer being together to being in the grand coalition. Any low risk player is indifferent between any arrangement that has $\nL$ low risk players and pays higher cost in any coalition with $\nL'<\nL$, so there is no set $S$ where the low risk players get strictly lower price. High risk players most prefer being with more low risk players and high risk players, so being in the grand coalition is their optimal arrangement. 
\end{proof}

Next, we investigate the \enquote{max} part of max-subsidy pricing: we show that any pricing scheme where the high risk players pay less than max-subsidy is one where the low risk players have an incentive to defect. 
\begin{lemma}
Assume a submodular cost function $\cost(\cd)$. Then, any pricing scheme where the high risk players pay less than max-subsidy (for any given coalition) is one where the low risk players have a group incentive to defect to a homogeneous pool of only low risk players. 
\end{lemma}
\begin{proof}
Suppose that the high risk players pay $\epsilon > 0$ less than the $\priceH(\cdot)$ max-subsidy price described above. Then, the total amount that the high risk players pay is: 
$$\cost(\{\nL',\nH' \}, \{\rL, \rH\}) - \cost(\{\nL',0 \}, \{\rL, 0\}) - \epsilon \cd \nH'$$
By efficiency, this means that the low risk players must pay: 
$$  \frac{\cost(\{\nL',0 \}, \{\rL, 0\})}{\nL'} + \epsilon \cd \frac{\nH'}{\nL'}$$
which is strictly greater than what they would pay in a group of $\nL'$ low risk players alone, so they have an incentive to defect. 
\end{proof}
This result shows that max-subsidy is the best we can do: it provides a tight lower bound on how much it is possible to subsidize the high risk group without destabilizing the pool. Bounds like this may be helpful for framing the scope of options available for subsidizing a certain group through market mechanisms, though it may be possible to provide stronger subsidizes through non-market mechanisms such as direct tax and subsidies.

\subsection{Implication for insurance application}
In the previous sections, we have shown results for a general cost function $\cost(\cd)$. In this section, we will translate these results into the insolvency-based pricing scheme. For conciseness, we will define $\RL = \rL \cd (1-\rL)$ and $\RH = \rH \cd (1-\rH)$. 

Applying the results of Lemma \ref{evenstab} tells us that an equal split pricing scheme is possible whenever: 
$$\frac{\V \cd \p{\NL \cd \rL + \NH \cd \rH + \bp \cd \sqrt{\NL \cd \RL + \NH \cd \RH}}}{\NL + \NH} $$
$$\leq \frac{\V\cd \p{\rL \cd \NL + \bp \cd \sqrt{\NL \cd \RL }}}{\NL}$$

Applying the results of Lemma \ref{maxsubpref} tells us that the max-subsidy pricing scheme in the insurance case is:
$$\priceL(\{\nL,\nH \}, \{\rL, \rH\}) = \V \p{\rL  + \bp \cd \frac{\sqrt{\RL}}{\sqrt{\nL}}}$$
$$\priceH(\{\nL,\nH \}, \{\rL, \rH\}) = \V\p{\rH + \bp \frac{\sqrt{\nL \RL + \nH \RH} - \sqrt{\nL \RL}}{\nH}} $$

Next, we consider the potential implications of these results on debates around fairness in insurance. As mentioned before, it is important to note that there are two countervailing forces at work: one is the additional costs that each person potentially brings to the pool; the other is the cost-savings that the collective enjoys by enlarging the pool and thereby reducing variance. This creates interesting dynamics: essentially, the submodularity of the cost function $\cost(\cd)$ produces extra \enquote{wiggle room.} In some cases, the benefits of increasing the size of the pool can swamp out the cost of including riskier participants. Under these circumstances, low-risk people are willing to let high-risk people join the pool --- or are willing to remain in the pool if high-risk people join --- if doing so results in lower prices for them, even if everyone is charged the same price. In particular, it seems that the case where solidarity might be easiest to achieve is where $\rH$ is not much larger than $\rL$ and $\NH$ is much larger than $\NL$. To put it simply, size ensures solidarity: the fact that the high-risk group is large enables the even sharing of costs. This runs counter to expectations because it might be reasonable to assume that a coalition needs solidarity --- a willingness to join together, even if it's not utility-maximizing for some --- to build a large coalition. However, the results indicate that the motivation works the other way around.

There are, of course, limits to how much insurers can save in costs by reducing variance and how much these savings can compensate for the difference in risks between groups. In fact, our model gives a precise account for how far we can pursue solidaristic goals before the pool begins to destabilize --- and thus when non-market mechanisms might be necessary to achieve those goals, such as taxation and redistribution policies. Our findings further demonstrate that imposing an even-split pricing scheme on the belief that it serves the goals of solidarity can have the opposite effect by discouraging people from remaining in the pool and actually drive up costs for all the remaining members, including the most price-sensitive. From this perspective, charging different prices --- something that actuarial fairness demands --- can have the effect of ensuring a greater willingness on the part of policyholders to stay in the pool --- that is, to be in solidarity with others and thus collectively enjoy the benefits of variance reduction. It is worth noting that in both cases, we analyzed the stability \emph{given that all insurers in a market follow the the same pricing rule}. This situation might occur in the even-split case if rules are mandated by the government that forbid price discrimination on certain characteristics \cite{avraham2014understanding}.

In this way, the results of this analysis can be helpful as a guidepost for debates around solidarity, explaining when it may be possible to achieve solidaristic goals while still ensuring a stable pricing scheme and when differential pricing can nevertheless serve solidaristic ends. In the next section, we will implement a similar analysis for actuarial fairness.

\section{Actuarial fairness under externalities of size}\label{Actfair}
In the previous section, we described pricing schemes whose goal was to minimize the cost paid by high risk individuals---a solidaristic goal. In this section, we will explore pricing schemes related to the actuarial fairness literature. As a reminder, common themes within this literature revolve around the desire for individuals to pay \enquote{their share} of what they contribute to overall costs. 

In the first subsection, we describe a few pricing schemes that might attempt to satisfy certain properties of actuarial fairness. However, these pricing schemes produce certain anti-social incentives for participants (that is, incentives that go against overall social welfare). In the next subsection, we explore two impossibility results indicating that these undesirable properties are actually necessary, if we wish to have other properties like efficiency. Finally, the last subsection considers the implications of these results in our insurance application.  

\subsection{Pricing schemes for insurance application}
In this subsection, we describe two different pricing schemes that might attempt to satisfy actuarial fairness properties. However, we also note that they have some undesirable properties. 

One well-known way of dividing costs is according to the Shapley value \cite{shapley1971cores}. For a cost-sharing game with $n$ participants and cost function $\cost(\cd)$, the Shapley value would assign a cost to player $i$ according to: 
$$\phi(i) = \frac{1}{n}\sum_{S \in [n]\backslash \{i\}}\frac{\cost(S \cup \{i\}) - \cost(S)}{\binom{n-1}{\vert S \vert}}$$
which can be interpreted as the average increase in cost player $i$ brings to a pool, where the average is taken over all possible pools of participants. It has been proven that the Shapley value is core-stable whenever the cost function is submodular \cite{shapley1971cores,bondareva1963some}, meaning that a pricing policy following the Shapley value is one where no subgroup is incentivized to leave the grand coalition. 

One drawback of the Shapley value is that it is computationally inefficient for large $n$, given that computing its value requires summing over all possible subsets of players. (Even if there are only two risk levels of players, there are still exponentially many ways that they can be arranged.)In this paper, we will use a related, but different pricing scheme we call \emph{proportional pricing}. With this pricing: 
$$\priceL(\{\nL,\nH \}, \{\rL, \rH\}) = \V \cd \p{\rL  + \bp \frac{\RL}{\sqrt{\RL \cd \nL + \RH \cd \nH}}}$$
$$\priceH(\{\nL,\nH \}, \{\rL, \rH\}) = \V \cd \p{\rH  + \bp \frac{\RH}{\sqrt{\RL \cd \nL + \RH \cd \nH}}}$$
It is straightforward to check that this satisfies efficiency. Note that this pricing scheme also has the property that adding another participant strictly reduces the premium any individual pays: for this reason, the grand coalition will minimize costs for both low risk and high risk players, which implies it is core-stable. 

However, this pricing scheme also has two other properties that may be undesirable. First, the price low risk players pay depends on the number and riskiness of the high risk players (and vice versa). This may be undesirable because, by the conception of actuarial fairness, a premium should depend solely on the risk that individual is responsible for. With the insolvency-based pricing, such a property no longer seems reasonable to aim for. In the next section, we show that, given some reasonable assumptions, it is impossible for a strictly submodular cost function $\cost(\cd)$ to give rise to a pricing scheme where players pay prices that are independent of the risks of other participants. It may be useful to note some irony in these results: to keep people from defecting from the insurance pool, which is the goal of actuarial fairness, the insurer needs to set individual premiums in a way that depends on the presence of other people in the pool, exactly what actuarial fairness forbids. 

Secondly, the proportional pricing scheme above has the property that players prefer that their partners are \emph{more} risky: for example, $\priceL(\cd)$ strictly decreases as $\rH$ increases. This produces the anti-social incentive to wish that the risk of other participants increases. Again, in the next section we show that such a property is also a necessary feature of a stable pricing scheme with a submodular pricing function.  
\subsection{Impossibility results}
First, we will show that it is impossible to have prices that are completely independent of the risk of other participants. 

\begin{lemma}\label{threequal}
The following three qualities cannot be achieved simultaneously: 
\begin{enumerate}
    \item Efficiency: $\cost(\cd )=\nL \cd\priceL(\cd ) + \nH \cd \priceH(\cd )$
    \item $\priceL(\{\nL, \nH\}, \{\rL, \rH\})$ is independent of $\rH, \nH$. 
    \item $\priceH(\{\nL, \nH\}, \{\rL, \rH\})$ is independent of $\rL, \nL$.
\end{enumerate}
\end{lemma}
\begin{proof}
In proving this, we will assume that the second and third properties hold, and use it to show a violation of efficiency. For conciseness, we will drop the number of players $\nL, \nH$, which are held constant in the equations, so that $\cost(\{\nL, \nH\}, \{\rL, \rH\}) = \cost(\rL, \rH) $. 

First, we consider a case where the high-risk player's cost is 0. By efficiency, we must have: 
$$\nL \cd \priceL(\rL, 0) + \nH \cd\priceH(\rL, 0)  = \cost(\rL, 0)$$ 
By property 2, $\priceH(\rL, 0) = \priceH(0, 0)$, so the equation can be rewritten as: 
$$\nL \cd \priceL(\rL, 0)  =  \cost(\rL, 0)- \nH \cd\priceH(0, 0)$$ 
Similarly, considering the case where the low risk player's cost is 0 gives the equation: 
$$\nH \cd \priceH(0, \rH)  =  \cost(0, \rH)- \nH \cd\priceL(0, 0)$$ 
We can then consider the case where both players have nonzero cost: we will show a violation of efficiency. The total amount that the players pay is
$$\nL \cd \priceL(\rL, \rH) + \nH \cd\priceH(\rL, \rH)$$
By properties 2 and 3 and previous equations: 
$$\nL \cd \priceL(\rL, \rH) = \nL \cd \priceL(\rL, 0) =\cost(\rL, 0)- \nH \cd\priceH(0, 0)$$
Similarly: 
$$\nH \cd \priceL(\rL, \rH) = \nH \cd \priceL(0, \rH) =\cost(0, \rH)- \nL \cd\priceL(0, 0)$$
We can use this to rewrite the equation as: 
$$\nL \cd \priceL(\rL, \rH) + \nH \cd\priceH(\rL, \rH)$$
$$ =\cost(\rL, 0)- \nH \cd\priceH(0, 0)+ \cost(0, \rH)- \nL \cd\priceL(0, 0)$$
We can drop some terms by considering the case where both players have 0 cost: 
$$\nL \cd \priceL(0, 0) + \nH \cd\priceH(0, 0)  = \cost(0, 0) = 0$$
which simplifies the sum down to $\cost(\rL, 0)+ \cost(0, \rH)$. 
In order for efficiency to hold, this implies that we must have:
$$ \cost(\rL, 0)+ \cost(0, \rH) = \cost(\rL, \rH)$$
As a reminder, $\cost(\rL,0)$ is independent of $\nH$. 
The last equation is saying that the cost associated with $\nL$ low risk individuals plus the cost associated with $\nH$ high risk individuals is equal to the cost associated with $\nL$ low risk individuals combined with $\nH$ high risk individuals --- which violates the fact that $\cost(\cd)$ is strictly submodular. 
\end{proof}

Next, we will show that there must exist some anti-social incentives: for example, any pricing scheme that is efficient and stable is one where players would wish their partners to have higher risk. We call this \enquote{anti-social} because it goes against broader social welfare: in this situation, some players would benefit from an increase in the true risk levels of the community.

\begin{lemma}\label{aligned}
For conciseness, we will again drop the number of players $\nL, \nH$, which are held constant in the equations. Assume that:  
$$\lim_{\rH\rightarrow 0}\br{\priceH(\rL, \rH)} = c$$
for some constant $c$. Then, it is not possible to have a pricing scheme that satisfies all three of the following properties: 
\begin{enumerate}
    \item Efficiency: $\cost(\cd )=\nL \cd\priceL(\cd ) + \nH \cd \priceH(\cd )$
    \item Aligned incentives: low risk players prefer that high risk players have lower risks on some interval including 0. 
    $$\frac{d}{d\rH}\priceL(\rL, \rH) >0 \ \forall \rH \in [0, r]$$
    \item Stability: for every level of risk $\rL, \rH$, both low risk players and high risk players benefit by being pooled together. 
\end{enumerate}
Property 2 is stated from the perspective of the low risk player, but the same logic (in the statement and proof) would work if it were stated from the perspective of the high risk player. 
\end{lemma}
Note that Lemma \ref{aligned} does not exclude the case where the low risk player's cost might stay constant as $\rH$ increases. However, Lemma \ref{threequal} shows that the low and high risk players cannot both have prices independent of each other's risk. 
\begin{proof}
In this lemma only, we will relax the assumption that $\rL < \rH$ so we are able to examine the limit as $\rH \rightarrow 0$ without requiring that $\rL \rightarrow 0$. For this proof, we assume that the first and second properties hold and use it to show that the third property cannot hold. First, we consider the case where $c \leq 0$. We start with property 1 (efficiency): 
$$ \cost(\rL, \rH)=\nL \cd\priceL(\rL, \rH) + \nH \cd \priceH(\rL, \rH)$$
Next, we take the limit as $\rH \rightarrow 0$ on both sides, which gives: 
$$ \cost(\rL, 0)=\nL \cd \lim_{\rH \rightarrow 0}\br{\priceL(\rL, \rH)} + \nH \cd c$$
Rearranging gives: 
$$\lim_{\rH \rightarrow 0}\br{\priceL(\rL, \rH)} = \frac{\cost(\rL, 0)}{\nL} - \frac{\nH}{\nL}c$$
So, as $\rH $ decreases towards 0, the price low risk players pay goes to something equal to or greater than $$\frac{\cost(\rL, 0)}{\nL} = \frac{\cost(\{\nL, 0\}, \{\rL, 0\})}{\nL}$$
which is the price they would pay in a group of only low risk players. Because of property 2, we know that $\priceL(\cdot)$ is \emph{decreasing} as $\rH$ is decreasing, so $\priceL()$ is strictly greater than the price the low risk players would pay if they were alone. This proves the statement in the case that $c\leq 0$. 

Next, we consider the case where $c>0$. In this case, we'll show that the high risk player has an incentive to defect. By assumption, 
$$\lim_{\rH\rightarrow 0}\br{\priceH(\rL, \rH)} = c > 0$$
However, we also have that:
$$\lim_{\rH\rightarrow 0}\br{\priceH(0, \rH)} = 0$$
So for some $\rH >0$, $\priceH(\cdot)$ is greater than the price it would get in a homogeneous group with only high risk players. Taken together, these cases prove the lemma. 
\end{proof}

\subsection{Implication for insurance application}
Finally, we will consider the implications of these results. First, these results tell us that independence is not the key to stability, as the arguments in favor of actuarial fairness would have us believe. In fact, allowing the premiums charged to one individual to be affected by the presence or absence of other people is how we are able to ensure a stable pool. A truly independent pricing scheme would \emph{over-charge} participants, which would be inefficient. With a submodular cost function, however, prices are being helpfully affected by the reduction in variance that occurs from a larger pool. 

Secondly, these results show that actuarial fairness misaligns incentives. Participants in the pool are incentivized to reduce their own risk --- but they are also incentivized to wish other members of the pool have \emph{higher} risk, which is not socially optimal. This could be seen as moral hazard, but one degree removed: wishing that other people would take on greater risk because it benefits you financially.

These results reveal the need to revisit the conceptual foundations of actuarial fairness in light of externalities of size, especially given that one of the main goals of actuarial fairness is to maintain a large pool of diversified risks --- the very thing that produces these externalities.

\section{Conclusion}
There are a few high-level results that are useful to take away from this work. Overall, a cost-sharing game with externalities of size is one where solidarity and actuarial fairness are not straightforward. The cost savings associated with larger groups can enable prices that strictly benefit both low and high cost individuals. More strongly than that, we give bounds on the lowest stable price we can give to the high cost group, as well as conditions for when even-split pricing is stable. Actuarial fairness also becomes more nuanced: the notions of independence and efficiency of pricing turn out to be at odds with each other. Additionally, requiring efficiency and stability actually produces its own kind of moral hazard. 

Our findings have broader implications for the FAccT community and for those concerned with issues of fairness in insurance. Our study highlights an important dynamic that has been overlooked in the FAccT literature: how we decide to treat one individual often depends on how we have decided to treat others, which is true in cases even beyond traditional ones like scarce resource allocation. This paper has focused on insurance, where the decision to offer insurance to one person at some price affects the terms on which we are able to offer insurance to others, but other domains exhibit similar properties. Existing work on fairness in machine learning has also not yet explored the value of economies of size and how they might ease the challenge of achieving certain equality-oriented notions of fairness. Economies of size help to align the interests of low-risk and high-risk populations and give us more room to maneuver when setting stable solidaristic pricing schemes. But they can also be of value in such domains as credit, where pooling default risk should have similar effects on the interest rates that lenders charge debtors.

More broadly, our results challenge common beliefs in insurance. We show that the stark distinction that people like to draw between solidarity and actuarial fairness in insurance falls apart upon closer examination, largely because people have failed to recognize what can be achieved with variance reduction. This finding calls for more flexible notions of fairness that are able to take into account dynamics produced by different cost functions --- including externalities of size, but also other variants on the cost function beyond what we have proposed in this paper.

\begin{acks}
We are grateful to Ian Ball, Hoda Heidari, Nicole Immorlica, Jon Kleinberg, and Manish Raghavan for extremely valuable discussions around earlier versions of this work. We would also like to thank the anonymous reviewers, researchers at the New York City lab of Microsoft Research, attendees at the NeurIPS workshop on Consequential Decisions in Dynamic Environments, and Zhuoya Fan for their helpful feedback. 
\end{acks}

\bibliographystyle{ACM-Reference-Format}
\bibliography{sample-base}

\clearpage
\appendix

\section{Proofs}\label{App:proofs}
\submod*
\begin{proof}
First, we define $S$ and $T$. Of low risk players, there are $\nLv{1}$ in set S but not set T, $\nLv{2}$ in set $T$ but not set $S$, and $\nLv{3}$ in both set S and set T. Similarly, for high risk players, there are $\nHv{1}$ in set S but not set T, $\nHv{2}$ in set $T$ but not set $S$, and $\nHv{3}$ in both set S and set T. 

For reference, the complete form of the cost function is repeated below: $\cost(\{\nL, \nH\}, \{\rL, \rH\}) = $
$$\V \p{\rL \cd \nL + \rH \cd \nH + \bp \cd \sqrt{\nL \cdot \rL \cd (1-\rL) + \nH \cd \rH \cd (1-\rH)}}$$

First, we can note that the $\V$ term is a constant: for simplicity, we can drop it. Next, we will show that the inequality holds for the linear component of $\cost(\cdot)$: in fact, it is an equality.

Focusing on the linear terms, the lefthand side of the inequality becomes: 
$$(\nLv{1} + \nLv{3}) \cd \rL + (\nHv{1} + \nHv{3}) \cd \rH  + (\nLv{2} + \nLv{3}) \cd \rL+ (\nHv{2} + \nHv{3}) \cd \rH $$
The righthand side becomes: 
$$(\nLv{1} + \nLv{2} + \nLv{3}) \cd \rL + (\nHv{1} + \nHv{2} + \nHv{3})\cd \rH + \nLv{3} \cd \rL + \nHv{3} \cd \rH$$
These sides are equal. 

Next, we look at the square root portion of $\cost(\cd)$. Again, the $\bp$ term is a constant that we drop for simplicity. For conciseness, we use the shorthand of $\RL = \rL \cd (1-\rL)$ and $\RH = \rH\cd (1-\rH)$. 

The lefthand side of the inequality becomes: 
$$\sqrt{(\nLv{1} + \nLv{3}) \cd \RL + (\nHv{1} +\nHv{3}) \cd \RH} $$
$$+ \sqrt{(\nLv{2} + \nLv{3}) \cd \RL + (\nHv{2} + \nHv{3}}) \cd \RH$$
The righthand side becomes: 
$$\sqrt{(\nLv{1} + \nLv{2} + \nLv{3}) \cd \RL + (\nHv{1} + \nHv{2} + \nHv{3}) \cd \RH} $$
$$+ \sqrt{\nLv{3} \cd \RL + \nHv{3} \cd \RH}$$
Next, we square both sides. The lefthand side becomes: 
$$ (\nLv{1} + \nLv{3}) \cd \RL + (\nHv{1} + \nHv{3}) \cd \RH + (\nLv{2} + \nLv{3}) \cd \RL + (\nHv{2} + \nHv{3}) \cd \RH$$
$$+ 2\sqrt{(\nLv{1} + \nLv{3}) \cd \RL + (\nHv{1} + \nHv{3}) \cd \RH} $$
$$\cd \sqrt{(\nLv{2} + \nLv{3}) \cd \RL + (\nHv{2} + \nHv{3}) \cd \RH}$$
The righthand side becomes: 
$$(\nLv{1} + \nLv{2} + \nLv{3}) \cd \RL + (\nHv{1} + \nHv{2} + \nHv{3}) \cd \RH + \nLv{3} \cd \RL + \nHv{3} \cd \RH$$
$$+ 2\sqrt{(\nLv{1} + \nLv{2} + \nLv{3}) \cd \RL + (\nHv{1} + \nHv{2} + \nHv{3}) \cd \RH}$$
$$\cd  \sqrt{\nLv{3} \cd \RL + \nHv{3} \cd \RH}$$
The terms without the square root are the same on each side, so we can drop them. Then, the inequality we are trying to show is:
$$2\sqrt{(\nLv{1} + \nLv{3}) \cd \RL + (\nHv{1} + \nHv{3}) \cd \RH} $$
$$\cd  \sqrt{(\nLv{2} + \nLv{3})\cd \RL + (\nHv{2} + \nHv{3}) \cd \RH}$$
$$\geq  2\sqrt{(\nLv{1} + \nLv{2} + \nLv{3}) \cd \RL + (\nHv{1} + \nHv{2} + \nHv{3}) \cd \RH}$$
$$\cd  \sqrt{\nLv{3} \cd \RL + \nHv{3} \cd \RH}$$
which is equivalent to showing: 
$$\p{(\nLv{1} + \nLv{3}) \cd \RL + (\nHv{1} + \nHv{3}) \cd \RH} $$
$$\cd \p{(\nLv{2} + \nLv{3}) \cd \RL + (\nHv{2} + \nHv{3}) \cd \RH}$$
$$\geq \p{(\nLv{1} + \nLv{2} + \nLv{3}) \cd \RL + (\nHv{1} + \nHv{2} + \nHv{3}) \cd \RH}$$
$$\cd  \p{\nLv{3} \cd \RL + \nHv{3} \cd \RH}$$
Expanding the lefthand side gives us: 
$$\RL^2 \cd (\nLv{1} + \nLv{3}) \cd (\nLv{2} + \nLv{3}) + (\nLv{1} + \nLv{3}) \cd (\nHv{2} + \nHv{3}) \cd \RL \cd \RH$$
$$ + (\nHv{1} + \nHv{3}) \cd (\nLv{2} + \nLv{3}) \cd \RL \cd \RH + (\nHv{1} + \nHv{3}) \cd (\nHv{2} + \nHv{3}) \cd \RH^2$$
Expanding out the righthand side gives: 
$$(\nLv{1} + \nLv{2} + \nLv{3}) \cd \nLv{3} \cd \RL^2 + (\nLv{1} + \nLv{2} + \nLv{3}) \cd \nHv{3} \cd \RL \cd \RH$$
$$(\nHv{1} + \nHv{2} + \nHv{3}) \cd \nLv{3} \cd \RH \cd \RL + (\nHv{1} + \nHv{2} + \nHv{3}) \cd \nHv{3} \cd \RH^2$$
The left and right side both have four terms. Focusing on the first term on each side, we can expand the lefthand side to get a coefficient on the $\RL^2$ term of: 
$$\nLv{1} \cd \nLv{2} + \nLv{1} \cd \nLv{3} + \nLv{3} \cd \nLv{2} + \nLv{3}^2$$
Expanding out the first term on the righthand side gives a coefficient of:  
$$\nLv{1} \cd \nLv{3} + \nLv{2} \cd \nLv{3} + \nLv{3}^2$$
The lefthand side contains every term on the righthand side, so it is greater than or equal to the righthand side. The inequality is strict whenever $\nLv{1}$ and $\nLv{2}$ are both strictly greater than 0. 

The fourth term on the lefthand side and the fourth term on the righthand side have a similar structure, and so the results are the same. The lefthand side is greater than or equal to the righthand side, with the inequality being strict whenever $\nHv{1}$ and $\nHv{2}$ are both strictly greater than 0. 

For the second and third terms on the left and righthand side, we expand and sum the terms. The lefthand side becomes: 
$$\nLv{1} \cd \nHv{2} + \nLv{1} \cd \nHv{3} + \nLv{3} \cd \nHv{2} + \nLv{3} \cd \nHv{3}$$
$$+ \nHv{1} \cd \nLv{2} + \nHv{1} \cd \nLv{3} + \nHv{3} \cd \nLv{2} + \nHv{3} + \nLv{3}$$
The righthand side becomes: 
$$\nLv{1} \cd \nHv{3} + \nLv{2} \cd \nHv{3} + \nLv{3} \cd \nHv{3}$$
$$+ \nHv{1} \cd \nLv{3} + \nHv{2} \cd \nLv{3} + \nHv{3} \cd \nLv{3}$$
All of the terms on the righthand side are also on the lefthand side, so the lefthand side is equal to or greater than the righthand side. The inequality is strict so long as either $\nLv{1} \cd \nHv{2}$ is strictly greater than 0 (both $\nLv{1}, \nHv{2}$ are strictly greater than 0) or $\nHv{1} \cd \nLv{2}>0$ (both $\nHv{1}, \nLv{2}$ are strictly greater than 0). 
So far, we have shown that 
$$\cost(S) + \cost(T) \geq \cost(S \cup T) + \cost(S \cap T)$$
and that the inequality is strict whenever at least one of the inequalities below holds: 
$$\nLv{1} \cd \nHv{2} > 0 \text{ or } \nHv{1} \cd \nLv{2} > 0$$
$$\text{ or } \nLv{1} \cd \nLv{2} > 0 \text{ or } \nHv{1} \cd \nHv{2} > 0$$
These conditions tell us that the inequality is strict whenever sets $S$ and $T$ both have strictly non-overlapping sections: neither is a subset of the other. 
\end{proof}

\evenstab*

\begin{proof}
For conciseness, in this proof we will drop the $\rL, \rH$ terms, which are constant throughout, so we write: 
$$\cost(\{\nL, \nH\}, \{\rL, \rH\}) = \cost(\nL, \nH)$$
First, we will show that the even-split pricing is decreasing in $\nL$. To do this, we can write the numerator as a telescoping sum equivalent to the cost of adding each term individually: 
$$\cost(\nL, \nH) = \sum_{j=1}^{\nH}\cost(0, i) - \cost(0, i-1)+\sum_{i=1}^{\nL}\cost(i, \nH) - \cost(i-1, \nH)$$
We claim that this is a sum of a decreasing sequence of terms. By submodularity, the sum over $j$ as the high risk players are added is decreasing. Also by submodularity, the sum over $i$ as the low risk players are added is decreasing. The last step we need to prove is that: 
$$\cost(1, \nH) -\cost(0, \nH)  < \cost(0, \nH) -\cost(0, \nH-1) $$
We can prove this by noting that, because $\rH > \rL$: $$\cost(1, \nH) -\cost(0, \nH) <\cost(0, \nH+1) -\cost(0, \nH)$$
$$<\cost(0, \nH) -\cost(0, \nH-1) $$
Then, we can view $\frac{\cost(\nL, \nH)}{\nL + \nH}$ as the average over a sequence of decreasing terms, which means that it is decreasing. 
Next, we will show that for any $\nH$ where the below inequality is satisfied, then the even-split price is decreasing in $\nH$:
$$\frac{\cost(\nL,\nH)}{\nL + \nH} \leq  \frac{\cost(\NL,0 )}{\NL }$$
To show this, first we can write out the numerator as a telescoping sum: 
$$\cost(\nL, \nH) = \sum_{i=1}^{\nL}\frac{\cost(\nL, 0)}{\nL}+\sum_{j=1}^{\nH}\cost(\nL, i) - \cost(\nL, i-1)$$
Note that $\frac{\cost(\nL, 0)}{\nL} \geq \frac{\cost(\NL, 0)}{\NL}  $ because $\nL \leq \NL$. 
Note that the sum over $j$ is a sequence of decreasing terms, again by submodularity. We know that the average of the entire sum is less than or equal to $\frac{\cost(\NL,0)}{\NL }$, which is equal to or smaller than each term within the $i$ sum over $\nL$. This must imply that there is at least one term in the sum over $j$ that is equal to or smaller than $\frac{\cost(\NL,0)}{\NL}$, which implies that the smallest term, $\cost(\nL, \nH) - \cost(\nL, \nH-1)$, must be at most $\frac{\cost(\NL,0)}{\NL }$. Because $\cost(\cd)$ is submodular:
$$\cost(\nL, \nH) - \cost(\nL, \nH-1)>\cost(\nL, \nH+1) - \cost(\nL, \nH)$$
and so increasing $\nH$ will add a new term that is smaller than any other term in the sum, thus decreasing the average.

These results, taken together, show the core-stability result: 
\begin{itemize}
    \item Players would not wish to go to any set $S$ with $\{\nL, \nH\}$ such that 
    $$\frac{\cost(\nL,\nH)}{\nL + \nH} \geq \frac{\cost(\NL,0)}{\NL}$$
    because this has equal or higher cost to the grand coalition. 
    \item For any set $S$ with the property that 
      $$\frac{\cost(\{\nL,\nH \}, \{\rL, \rH\})}{\nL + \nH} < \frac{\cost(\{\NL,0 \}, \{\rL,0\})}{\NL }$$
      the cost can be strictly decreased by increasing $\nH$, which would imply that the grand coalition has lower cost. 
    \item For any set where $\nH = \NH$ and so cannot be increased, we know that we can strictly decrease the cost by increasing $\nL$, again implying that the grand coalition has lower cost. 
\end{itemize}

\end{proof}

\maxsubpref*

\begin{proof}
For conciseness, in this proof we will drop the $\rL, \rH$ terms, which are constant throughout, so we write: 
$$\cost(\{\nL, \nH\}, \{\rL, \rH\}) = \cost(\nL, \nH)$$
First, we wish to show that the function below decreases as $\nL$ increases: 
$$ \frac{\cost(\nL,0)}{\nL}$$
We can rewrite the total cost as the sum of the marginal costs: 
$$\cost(\nL,0) = \sum_{i=1}^{\nL}\br{\cost(i,0) - \cost(i-1,0)}$$
An equivalent definition of submodularity is: 
$$\cost(S \cup \{j\}) - \cost(S) \geq \cost(T \cup \{j\}) - \cost(T) \text{ for } S \subset T \subset N\backslash \{j\}$$
which implies that the sum of terms above are decreasing as $i$ increases. We can then view $\frac{\cost(\nL,0)}{\nL}$ as the average of a sequence of decreasing numbers, which decreases as $\nL$ increases. 

Next, we wish to show that the price high risk players pay decreases with both $\nL$ and $\nH$. To show that this is true for $\nL$, we first note that by the definition of submodularity, 
$$\cost(\nL,\nH) - \cost(\nL,0) > \cost(\nL+1,\nH) - \cost(\nL+1,0)$$
So as $\nL$ increases, the numerator of high risk player's cost decreases and the denominator stays the same, so the high risk player's overall cost decreases. 
Next, we will show that the high risk player's cost is decreasing in $\nH$. For conciseness, we define: 
$$h(\nL,\nH) = \cost(\nL,\nH) - \cost(\nL,0)$$
which is a function representing the marginal cost of the $\nH$ high risk players. Note that we can write:
$$h(\nL,\nH) =h(\nL,\nH) - h(\nL,0) $$
which is true because: 
$$h(\nL,0) = \cost(\nL,0) - \cost(\nL,0 )=0$$
Using the same trick as before, we can then write that this is equivalent to: 
$$= \sum_{i=1}^{\nH}\br{h(\nL,i) - h(\nL,i-1)}$$
The difference $h(\nL,i) - h(\nL,i-1)$ represents the marginal cost of the $i$th high risk player. Because $\cost(\cd)$ is a submodular function, this marginal cost is decreasing. Then, the price high risk players pay is $\frac{h(\nL,i)}{\nH}$, which we can view as an average of a sequence of decreasing numbers, which decreases as $\nH$ increases. 
\end{proof}

\end{document}